\documentclass[conference]{IEEEtran}

\usepackage{graphicx}
\usepackage{subfig}

\usepackage{color} 
\usepackage{url}
\usepackage{amsmath}
\usepackage{amsthm}
\usepackage{amssymb}

\numberwithin{equation}{section}
\newtheorem{theorem}{Theorem}[section]
\newtheorem*{theorem*}{Theorem}

\newtheorem{lemma}[theorem]{Lemma}
\newtheorem{conjecture}[theorem]{Conjecture}
\newtheorem{corollary}[theorem]{Corollary}
\theoremstyle{definition}

\newcommand{\R}{\mathbb{R}}

\newcommand{\C}{\mathbb{C}}

\newcommand{\M}{\mathcal{M}}

\DeclareMathOperator{\GL}{GL}

\DeclareMathOperator{\SO}{SO}

\DeclareMathOperator{\Hom}{Hom}

\newcommand{\GLV}{\mathrm{GL}(V)}

\IEEEoverridecommandlockouts 

\begin{document}

\title{The generalized phase retrieval problem over compact groups}

\author{
    \IEEEauthorblockN{Tamir Bendory and Dan Edidin}
    \thanks{T.B. is with the School
of Electrical Engineering, Tel Aviv University, Israel. D.E. is with the Department of Mathematics,  University of Missouri, MO, USA. 
    T.B. and D.E. were supported by the BSF grant no. 2020159. T.B. is also supported in part by the NSF-BSF grant no. 2019752, in part by the ISF grant no. 1924/21, and in part by a grant from The Center for AI and Data Science at Tel Aviv University (TAD). D.E. was also supported by NSF-DMS 2205626.}
}

\maketitle
 

\begin{abstract}
The classical phase retrieval problem involves estimating a signal from its Fourier magnitudes (power spectrum) by leveraging prior information about the desired signal. This paper extends the problem to compact groups, addressing the recovery of a set of matrices from their Gram matrices. In this broader context, the missing phases in Fourier space are replaced by missing unitary or orthogonal matrices arising from the action of a compact group on a finite-dimensional vector space. This generalization is driven by applications in multi-reference alignment and single-particle cryo-electron microscopy, a pivotal technology in structural biology.
We define the generalized phase retrieval problem over compact groups and explore its underlying algebraic structure. We survey recent results on the uniqueness of solutions, focusing on the significant class of semialgebraic priors. Furthermore, we present a family of algorithms inspired by classical phase retrieval techniques. Finally, we propose a conjecture on the stability of the problem based on bi-Lipschitz analysis, supported by numerical experiments.
\end{abstract}

\section{Introduction}

\textbf{The classical phase retrieval problem.} 
The phase retrieval problem originated in the early 20th century with X-ray crystallography, a key method for determining molecular structures that has driven significant scientific advancements. 
The problem is generally formulated as:  
\begin{equation} \label{eq:pr}  
    \text{find } x \in \Omega \quad \text{subject to} \quad y = |Fx|^2,  
\end{equation}  
where \( y \in \mathbb{R}^M_{\geq 0} \) is the measurement vector, \( F \in \mathbb{C}^{M \times N} \) is a Fourier-type matrix, \( \Omega \) is the signal space of interest, and the absolute value is taken entry-wise. The signal space \( \Omega \) varies by application: in X-ray crystallography, it represents sparse signals~\cite{elser2018benchmark}; in other cases, \( x \) may be limited to a known support~\cite{barnett2022geometry}. In ptychography, \( \Omega \) is the column space of a short-time Fourier transform matrix~\cite{jaganathan2016stft, bendory2017non, iwen2020phase}. Phaseless measurements are typically invariant under symmetry groups determined by \( F \) and \( \Omega \), allowing only the orbit of \( x \) under this symmetry to be uniquely recovered. For a detailed survey, see~\cite{shechtman2015phase, bendory2017fourier, grohs2020phase, fannjiang2020numerics, bendory2022algebraic}.  

\textbf{Multi-reference alignment (MRA) and cryo-EM.}
To generalize phase retrieval to compact groups, we introduce the multi-reference alignment (MRA) problem~\cite{bandeira2014multireference}. Let \( G \) be a compact group acting on a vector space \( V \). Each MRA observation \( y \) is modeled as:  
\begin{eqnarray} \label{eq:mra}  
    y = g \cdot x + \varepsilon,  
\end{eqnarray}  
where \( g \in G \), \( \varepsilon \sim \mathcal{N}(0, \sigma^2 I) \) is Gaussian noise independent of \( g \), \( \cdot \) denotes the group action, and \( x \in V \) is the signal of interest.  We assume that \( g \) is uniformly distributed over \( G \) with respect to the Haar measure. The objective is to estimate the signal \( x \in V \) from \( n \) realizations given by:  
\begin{equation}  
    y_i = g_i \cdot x + \varepsilon_i, \quad i = 1, \ldots, n.  
\end{equation}

The results presented in this paper apply to any MRA model where a compact group \( G \) acts on a finite-dimensional space~\( V \). 
Notable examples include the group of two-dimensional rotations, \( \mathrm{SO}(2) \), acting on band-limited images~\cite{bandeira2020non, janco2022accelerated}, and the group of three-dimensional rotations, \( \mathrm{SO}(3) \), acting on band-limited signals defined on the sphere~\cite{bandeira2023estimation}. The latter occurs in cryo-electron microscopy and cryo-electron tomography~\cite{bendory2020single,watson2024advances}.

\textbf{The second moment of MRA.}
The underlying idea of the classical method of moments is that, with sufficient samples, the moments of the observed data can be used to approximate the corresponding moments of the unknown signal accurately. Recent works have focused on the second moment of the MRA model, expressed as (a bias term is omitted):  
\begin{equation}  
    \mathbb{E}[yy^*] = \int_G (g \cdot x)(g \cdot x)^* \, dg \approx \frac{1}{n} \sum_{i=1}^n y_i y_i^*.  
\end{equation}  
When \( n = \omega(\sigma^4) \), the right-hand side almost surely provides an accurate approximation of the second moment.  

There are several reasons to focus on the second moment, motivating this paper. First, it has been shown that the highest-order moment necessary to recover a signal determines the sample complexity of the MRA model in high-noise regimes~\cite{abbe2018estimation}. Since generic signals can often be determined from the third moment, this leads to a sample complexity of \( \omega(\sigma^6) \)~\cite{bandeira2023estimation, bendory2017bispectrum, perry2019sample}. Consequently, the number of samples required to achieve a desired error increases rapidly with noise levels. This observation motivates identifying classes of signals that can be recovered using only the second moment, for which \( \omega(\sigma^4) \) samples suffice for accurate recovery.  
In addition, the second moment has a lower dimension than the higher moments, reducing storage requirements and computational complexity.  
Finally, as introduced in~\cite{bendory2023autocorrelation} and explained later in this paper, the algebraic structure of the second moment enables a natural generalization of effective phase retrieval algorithms.  

\textbf{Cryo-EM.}
Single-particle cryo-electron microscopy (cryo-EM) is a groundbreaking technology in structural biology, enabling the elucidation of the structure and dynamics of biological molecules. 
Under certain simplifications, cryo-EM observations can be modeled as~\cite{bendory2020single, singer2020computational}:  
\begin{eqnarray} \label{eq:cryoEM}  
    y = P(g \cdot x) + \varepsilon,  
\end{eqnarray}  
where \( P \) is a tomographic projection, \( G \) represents the group of 3D rotations, and \( x \) denotes the electrostatic potential of the molecule.  Remarkably, the second moment of~\eqref{eq:cryoEM} is invariant to the tomographic projection~\cite{kam1980reconstruction, bendory2024sample}. Therefore, from the perspective of the second moment analysis discussed in this paper, the cryo-EM model can be viewed as a special case of the MRA model~\eqref{eq:mra}.

\textbf{This paper.} 
This paper's primary objective is to introduce the generalized phase retrieval problem over compact groups by analyzing the second moment in the MRA problem. As special cases, we provide a more detailed discussion of two key applications in structural biology: X-ray crystallography and cryo-EM.  
Through the lens of representation theory, we demonstrate in Theorem~\ref{thm:second_moment} and Corollary~\ref{cor:ambiguity} that the second moment determines the signal up to a set of unitary matrices. The dimension of this set is governed by the decomposition of the signal space into irreducible representations of the group.  
To recover the missing unitary matrices, we focus on the role of semialgebraic priors, encompassing crucial signal processing priors such as linear priors, sparsity, and deep generative models. 
We state a general transversality result for intersections of semialgebraic sets with orbits of compact groups, Theorem~\ref{thm:semialgebraic}, and explore its implications for specific cases, such as phase retrieval and cryo-EM. Finally, we show how classical phase retrieval algorithms can be naturally extended to compact groups and propose a conjecture regarding the stability of the problem.

\section{The generalized phase retrieval problem} \label{sec:second_moment}

\subsection{The second moment of the MRA problem}

A general finite-dimensional representation of a compact group can be decomposed as 
\begin{equation} \label{eq:V}
	V = \bigoplus_{\ell = 1}^L V_\ell^{\oplus R_\ell},
\end{equation}
where \( V_\ell \) are distinct (non-isomorphic) irreducible representations of \( G \), each with dimension \( N_\ell \). 
An element \( x \in V \) has a unique \( G \)-invariant decomposition as a sum:
\begin{equation} \label{eq.decomp}
	x = \sum_{\ell = 1}^L \sum_{i=1}^{R_\ell} x_\ell[i],
\end{equation} 
where \( x_\ell[i] \) belongs to the \( i \)-th copy of \( V_\ell \).

Let $X_\ell\in\C^{N_\ell\times R_\ell}$ be  a matrix whose columns are $x_\ell[i]$.
In~\cite{bendory2024sample}, we have proven the following result. 
\begin{theorem} \label{thm:second_moment}
   The second moment of the MRA model~\eqref{eq:mra} provides  the tuple of Gram matrices $X_\ell^*X_\ell\in\C^{R_\ell\times R_\ell}$ for $\ell=1,\ldots,L$.
   \end{theorem}
   \begin{proof}[Sketch of the proof] For $x \in V$ and $g \in G$, the outer product $(g \cdot x) (g \cdot x)^*$ can be identified with the tensor $(g \cdot x) \otimes \overline{(g \cdot x)} \in V \otimes V^*$. Averaging over the group $G$ means that the second moment is an invariant element of $V \otimes V^*$, or, equivalently,
   a $G$-invariant element of $\Hom(V,V)$. By definition, a $G$-invariant element $\phi \in \Hom(V,V)$ is a $G$-equivariant linear transformation $V \to V$ meaning that 
   $\phi(g \cdot v) = g \cdot \phi (v)$ for any $v \in V$.
  By Schur's lemma, if $V_\ell$ and $V_m$ are two irreducible representations of $G$, then a $G$-equivariant linear transformation $V_\ell \to V_m$ is zero if $\ell \neq m$ and a scalar multiple of the identity if $\ell = m$. This implies that the second moment can be viewed
   as a linear transformation
   $\sum_{\ell =1}^L \sum_{i,j = 1}^{R_\ell} c_{i,j}I_{\ell,i,j}$, where $I_{\ell,i,j}$ denotes the identity operator from the $i$-th copy of $V_\ell$ to the
   $j$-th copy of $V_\ell$. A trace calculation shows that
   $c_{i,j} = \frac{\langle x_\ell[i], x_\ell[j] \rangle}{N_\ell}$. Hence, the second moment determines the
   Gram matrices $X_\ell^* X_\ell$.
\end{proof}
The following is an important consequence of Theorem~\ref{thm:second_moment}.
\begin{corollary}\label{cor:ambiguity}
    	Let $V$ be of the form~\eqref{eq:V}. Then, a vector $x\in V$ is determined from the second moment up to the action
	of the ambiguity group $H = \prod_{\ell = 1}^L{U(N_\ell)}$.
\end{corollary} 

Recovering the missing unitary matrices is called the \emph{generalized phase retrieval problem}.

\subsection{Examples} 
\emph{Phase retrieval.} 
Consider the group \( G = \mathbb{Z}_N \) acting on \( \mathbb{C}^N \) by cyclic shifts. In the Fourier domain, the cyclic group \( G = \mathbb{Z}_N \) acts by multiplication by roots of unity. Specifically, we identify \( \mathbb{Z}_N \) with \( \mu_N \), where \( \mu_N \) denotes the \( N \)-th roots of unity. If \( \omega \in \mu_N \), then 
\begin{equation*}\label{eq:fourier_action}
	\omega \cdot (x[0], \ldots, x[N-1]) = (x[0], \omega x[1], \ldots, \omega^{N-1} x[N-1]).
\end{equation*}
The vector space \( \mathbb{C}^N \), under this action of \( \mu_N \), decomposes into a sum of one-dimensional irreducible representations, with \( N_\ell = R_\ell = 1 \), so that \( N = L \).
The second moment of a vector \( x \in \mathbb{C}^N \) in the Fourier domain is the power spectrum \( (|x[0]|^2, \ldots, |x[N-1]|^2) \). This determines the vector up to the action of the group \( (S^1)^N \), since \( U(1) = S^1 \).

\emph{Cryo-EM.}
Let \( L^2(\mathbb{R}^3) \) be the Hilbert space of complex-valued \( L^2 \) functions on \( \mathbb{R}^3 \). In cryo-EM, we are interested in the action of \( SO(3) \) on the subspace of \( L^2(\mathbb{R}^3) \) corresponding to the Fourier transforms of real-valued functions on \( \mathbb{R}^3 \), which represent the Coulomb potential of an unknown molecular structure.
Using spherical coordinates \( (r, \theta, \phi) \), we consider a finite-dimensional approximation of \( L^2(\mathbb{R}^3) \) by discretizing \( x(r, \theta, \phi) \) with \( R \) samples \( r_1, \ldots, r_{R} \) of the radial coordinates and bandlimiting the corresponding spherical functions \( x(r_i, \theta, \phi) \). This is a standard assumption in the cryo-EM literature (e.g.,~\cite{bandeira2020non}).
Mathematically, this means that we approximate the infinite-dimensional representation \( L^2(\mathbb{R}^3) \) with the finite-dimensional representation
$V = \left(\bigoplus_{\ell = 0}^L V_\ell\right)^R,$
where \( L \) is the bandlimit, and \( V_\ell \) is the \( (2\ell + 1) \)-dimensional irreducible representation of \( SO(3) \), corresponding to harmonic polynomials of frequency \( \ell \). 
An orthonormal basis for \( V_\ell \) is the set of spherical harmonic polynomials \( \{Y_\ell^m(\theta, \phi)\}_{m = -\ell}^\ell \) and the dimension of this representation is \( R(L^2 + 2L + 1) \).

Viewing an element of $V$ as a radially discretized function on $\R^3$, 
we can view
$x \in V$ as an $R$-tuple
$x = (x[1], \ldots , x[R]),$
where
$x[r] \in L^2(S^2)$ is an $L$-bandlimited function.
Each $x[r]$ can be expanded in terms of the basis functions $Y_\ell^m(\theta, \varphi)$ as follows 
\begin{equation} \label{eq.function}
	x[r] = \sum_{ \ell=0}^L\sum_{m=-\ell}^\ell X_{\ell}^m[r]
	Y_{\ell}^m.
\end{equation}
Therefore, the problem of determining a structure reduces to determining the unknown coefficients $X_\ell^m[r]$ in \eqref{eq.function}.
Since the tomographic projection in the cryo-EM model does not affect the second moment~\cite{kam1980reconstruction,bendory2024sample}, we can invoke Theorem~\ref{thm:second_moment} to conclude that the second moment determines the matrices
\begin{equation*} \label{eq:Bl}
	X_\ell^* X_\ell, \quad  \ell = 0, \ldots L, \quad X_{\ell} = \left(X_{\ell}^m[r_i]\right)_{m=-\ell, \ldots , \ell , i = 1, \ldots R}.
\end{equation*}

\section{Semi-algebraic priors}

 A semialgebraic set $\M \subseteq \R^N$ is a finite union of sets defined by polynomial equality and inequality constraints. 
Any semialgebraic set~$\M$ can be written as a finite union of smooth manifolds, and the dimension of~$\M$ is defined as the maximal dimension of these manifolds.
The assumption that a signal lies in a semialgebraic set is called a \emph{semialgebraic prior}. This work is motivated by three important special cases of semialgebraic sets.

\begin{itemize}
    \item \emph{Linear priors.} The assumption that the signal lies in some low-dimensional subspace, as in PCA.

\item \emph{Sparse priors.} The assumption that the signal can be represented by only a few coefficients under suitable dictionary~\cite{elad2010sparse}.

\item \emph{Deep generative models.} The ubiquitous assumption that the signal lies in the image of a neural network of the form
$x=A_\ell\circ \eta_{\ell-1} \circ A_{\ell-1}\circ \ldots \circ \eta_1 \circ A_1(z),$
where $z$ resides in a low dimensional (latent) space, the $A_i$'s are affine transformations, and the $\eta_i$'s are semialgebraic activation functions (e.g., ReLU). 
\end{itemize}

\subsection{Transversality theorem}
We state a transversality theorem for semialgebraic sets in orthogonal representations \( V \) of a compact group \( H \). We provide bounds on the dimension of a generic semialgebraic set \( \mathcal{M} \), ensuring that it is transverse to the orbits of the group action. In other words, the \( H \)-orbit of any point \( x \in \mathcal{M} \) intersects \( \mathcal{M} \) only at \( x \).
This implies that if \( \Psi \) is a measurement function (such as the second moment) that separates \( H \)-orbits, then \( \Psi \) is one-to-one when restricted to \( \mathcal{M} \). Consequently, the prior knowledge that \( x \) lies in \( \mathcal{M} \) ensures that \( x \) is uniquely determined by the measurement \( \Psi(x) \).

The following result is formulated in terms of two parameters: the dimension of the semialgebraic set \( \mathcal{M} \) and the effective dimension of the representation \( K \), defined as the dimension of the representation minus the maximum dimension of the orbits:
\begin{equation} \label{eq:K}
    K = \dim V - k(H),
\end{equation}
where 
\( k(H) = \max_{x \in V} \dim Hx \). 

Clearly, the problem becomes easier when \( \mathcal{M} \) is of a lower dimension (a smaller semialgebraic set) and \( K \) is larger (a greater effective dimension of the representation). We prove that the gap between \( K \) and \( \mathcal{M} \) can be small. A detailed formulation of the theorem (which also holds for unitary representations) can be found in~\cite{bendory2024transversality}. When we say that a transversality statement holds for \( \mathrm{GL}(V) \)-generic semialgebraic sets of a given dimension, we mean that if \( \mathcal{M} \) is \emph{any} semialgebraic set of the given dimension, then the transversality statement holds for the translate of that semialgebraic set by a generic linear transformation \( A \in \mathrm{GL}(V) \).

\begin{theorem}\label{thm:semialgebraic}
Let $V$ be an orthogonal representation of a compact Lie group $H$. Let $\M \subseteq V$ be a $\GLV$-generic semi-algebraic set of dimension $M$,
and 
let $K$ be as in~\eqref{eq:K}.
If $K>M$, then for a generic $x \in \M$ if $h \cdot x \in M$ for some $h \in H$, then 
 $h \cdot x = \pm x$. If $K>2M$, then it holds for all $x$.
\end{theorem}

\begin{proof}[Sketch of the proof]
For each pair of vectors $x,y \in V$, set
$\GLV(x,y) = \{ A\in \GLV| A \cdot x \in H( A\cdot y)\}.$ With this notation, we say that
$x \sim y$ if $\GLV(x,y) =\GLV$ meaning that $A \cdot x$
and $A \cdot y$ lie in the same $H$ orbit for every $A\in \GLV$. 
\begin{lemma}
If $x \not \sim y$, then 
$\dim \GLV(x,y) \leq \dim \GLV - K$. 
\end{lemma}
\begin{proof} If $x, y$ are parallel, then $Ax$ and $Ay$ are parallel for any $A \in \GLV$. Since $H$ acts by orthogonal 
transformations the fact that $Ax$ is in the $H$-orbit
of $Ay$ means that $|Ax| = |Ay|$ so $x,y$ differ by a sign
and are therefore equivalent. Thus, we may assume that
$x$ and $y$ are linearly independent.
Choose an ordered basis $b_1 = x, b_2 = y, \ldots b_N$ for $V$. The matrix of $A \in \GLV$ with respect to this basis has first row $Ax$ and second row
$Ay$. The condition that $Ax$ is in the orbit of $Ay$ implies
that the vector $Ax$ lies in a real algebraic subset
of $V$ of dimension at most $k(H)$.
\end{proof}
Theorem~\ref{thm:semialgebraic} then follows from the following result, which is a special case of the {\em Fiber Lemma}~\cite[Lemma 6.1]{bendory2024transversality}.
\begin{lemma} \label{lemma:fiber}
Assume that $\dim \GL(x,y) \leq \dim \GLV - K$ for all $x \not \sim y$.
Then, for a generic linear transformation $A \in \GLV$, the following hold, where $A(\M)$ denotes the translate of $\M$ by $A$.
\begin{enumerate}
		\item If $K > M$,
		then for a generic vector $x \in A(\M)$  if $y = h \cdot x 
  \in A(\M)$ for some $h \in H$, then $x \sim y$. 
		\item If $K > 2M$, 
 then for any vector $x \in A(\M)$  if $y = h \cdot x 
  \in A(\M)$ for some $h \in H$, then $x \sim y$. 
  	\end{enumerate}
  \end{lemma}
\end{proof}

\begin{corollary}\label{cor:signal_recovery}
Let $V$ be a real representation of a compact group $G$ of the form~\eqref{eq:V} and let $H = \prod_{\ell =1}^L O(N_\ell)$. Let $\M \subseteq V$ be a $\GLV$-generic semi-algebraic set of dimension $M$, and  
$K = \dim V - k(H)$ as given in~\eqref{eq:K}. 
Then, if $K>M$, a generic vector $x \in \M$ is determined (up to a possible sign) by the second moment of~\eqref{eq:mra} with respect to $G$. If $K>2M$, then this holds for every $x$.
\end{corollary}

\subsection{Examples} 
We now present the applications of Corollary~\ref{cor:signal_recovery}. 

\begin{corollary}[Phase retrieval] \label{cor:phase_retrieval} 
Consider the phase retrieval problem of recovering a signal $x\in\R^N$ from its power spectrum. 
If $\M$ is a 
$\GLV$-generic semialgebraic set of dimension $M$ with $N \geq 2M$, then the generic vector $x \in \M$ is
determined (up to a sign) from its power spectrum. If $N\geq 4M$, then this holds for every vector.
\end{corollary}

For the cryo-EM model, if we assume that $R \geq 2L+1$, then the orbits of the
action of $H = \prod_{\ell =0}^L O(2\ell +1)$ have full dimension,
so $k(H) =\sum_{\ell=0}^L\binom{2\ell +1}{2}\approx\frac{2L^3}{3}.$ A precise statement of the following corollary is provided in~\cite{bendory2024transversality}.

\begin{corollary}[Cryo-EM] \label{cor:cryoEM}
Consider the cryo-EM model described above, 
where the signal is taken from the representation 
 $V = \oplus_{\ell =0}^{L} V_\ell^{\oplus R}$ of $\SO(3)$ with $R \geq 2L+1$. Let $\M$ be a $\GLV$-generic semi-algebraic subset of dimension $M$ and 
let $K\approx L^2\left(R+\frac{2L}{3}\right).$ 
If $K>M$, then a generic $x \in \M$ is determined (up to a sign) by its second moment.
Likewise,  if $K>2M$, then this holds for every $x$.
   \end{corollary}

\section{Algorithms}
\textbf{Classical phase retrieval algorithms.}
The most effective phase retrieval algorithms are based on two projection operators. In X-ray crystallography, \( P_1 \) projects onto the best sparse signal approximation, while \( P_2 \) projects onto signals matching the measured power spectrum. In other contexts, \( P_1 \) adapts based on the prior; for instance, if the signal's support is known, \( P_1 \) zeros out values outside this support. The most straightforward algorithm, alternating projection, iterates as \( x \rightarrow P_2P_1(x) \). More advanced algorithms have since emerged, such as RAAR, HIO, and RRR~\cite{elser2003phase, luke2004relaxed, fienup1982phase}. While the properties of these algorithms for non-convex problems remain unclear~\cite{elser2018benchmark}, these algorithms effectively tackle complex computational challenges~\cite{elser2007searching}. 

\textbf{Generalizing phase retrieval algorithms.} 
Classical phase retrieval algorithms can be extended for signal recovery from the second moment of MRA and cryo-EM models, as demonstrated in~\cite{bendory2023autocorrelation}. This involves generalizing \( P_2 \) from retrieving the set of missing phases to unitary matrices. The concept is relatively straightforward: suppose we are given the Gram matrix \( X_{\ell}^T X_{\ell} \) and a current estimate of the matrix \( \tilde{X}_{\ell} \). The projection of \( \tilde{X}_{\ell} \) onto the set of matrices with the same Gram matrix is obtained by solving the Procrustes problem:
$\arg\min_{O_\ell} \|\tilde{X}_{\ell} - O_\ell X_{\ell} \|_{\text{F}}^2,$ where $O_\ell$ is an orthogonal matrix.
This problem can be efficiently solved using SVD.
The algorithm alternates between projecting the signal onto constraints defined by priors and by the second moment of the observable images. This approach generalizes crystallographic phase retrieval to any compact group and allows for the integration of additional prior knowledge through efficient projection operators. 

\textbf{Numerical experiments.}
To numerically examine the behavior of the proposed algorithm, we conducted the following experiments. We generated a matrix \( X \) of size \( 8 \times 4 \) that lies in a low-dimensional subspace of dimension $K$ and whose entries are drawn independently of a normal distribution. The objective is to recover \( X \) from its Gram matrix \( X^T X \in \mathbb{R}^{4 \times 4} \) using the alternating projection algorithm and based on the linear prior, starting from an initial guess drawn from the same distribution. 
We report the median error to account for potential convergence issues arising from the non-convexity of the problem. The results are presented in Figure~\ref{fig:experiments}.

The first experiment investigates the number of iterations required to estimate the matrix \( X \) up to a sign. Specifically, we counted the number of iterations needed to achieve a normalized error smaller than \( 10^{-6} \) as a function of \( K \), with the number of iterations capped at 1000. This experiment was repeated 10,000 times for each value of \( K \).
The second experiment examines the normalized recovery error as a function of the noise level. In this case, \( K = 10 \), and we observe \( (X+ \eta)^T(X + \eta) \), where \( \eta \) is a noise matrix with i.i.d.\ normal entries, zero mean, and variance \( \sigma^2 \). From this noisy Gram matrix, we aimed to estimate the ground truth \( X \). 
The recovery error increases smoothly with the noise level.


\begin{figure}[ht]
    \centering
    \subfloat[\label{fig:iterations}]{
    \includegraphics[width=0.45\linewidth]{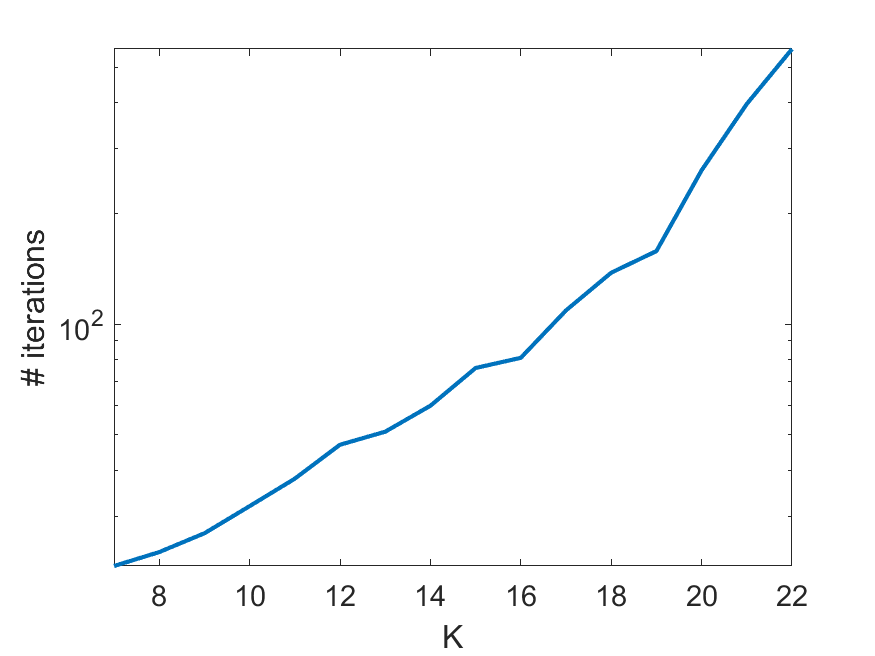}
    }
    \hfill
    \subfloat[\label{fig:error_noise}]{
    \includegraphics[width=0.45\linewidth]{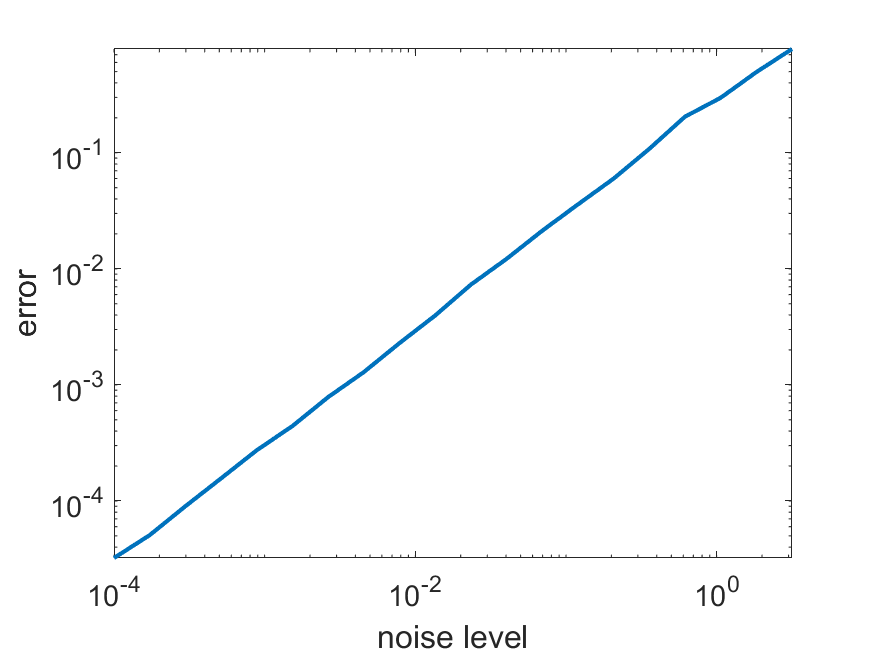}
    }
    \caption{\label{fig:experiments}In both experiments, the sought matrix is of dimension $4\times 8$. (a) The number of iterations required to achieve a recovery error smaller than $10^{-6}$ as a function of the dimensionality of the subspace $K$. (b) The recovery error as a function of the noise level when the matrix lies in a subspace of dimension $K=10$.}
    \label{fig:side_by_side_two_columns} \vspace{-15pt}
\end{figure}



\section{Future directions: bi-Lipschitz analysis}
The numerical results above suggest that, under linear priors, the recovery problem is well-conditioned. This, along with recent results~\cite{balan2016lipschitz, derksen2024bi}, leads us to conjecture that the recovery problem is bi-Lipschitz. Let \( V \) be a real representation of a compact group \( G \) that decomposes as in~\eqref{eq.decomp}, and let \( \mathcal{M} \) be a linear subspace such that the second moment is injective up to a sign. We propose the following conjecture, which generalizes the result of~\cite{balan2016lipschitz} to dimensions greater than one.

\begin{conjecture} \label{conj:bilipschitz}
The map
$ \M \to \R^N = \prod_{i=1}^L\R^{N_\ell \times R_\ell}$
given by 
$$(X_1, \ldots , X_L) \mapsto (\sqrt{X_1^*X_1}, \ldots , \sqrt{X_L^*X_L})
$$
is bi-Lipschitz when $\M$ is given the pseudo-metric
$\rho(x,y) = \min | x \pm  y|$
and $\R^N = \prod_{i=1}^L \R^{N_\ell \times R_\ell}$ is given the Euclidean metric.
\end{conjecture}

\bibliographystyle{plain}

\end{document}